\documentclass[10pt,english]{IEEEtran}
\usepackage[T1]{fontenc}
\usepackage[latin9]{inputenc}
\usepackage{geometry}
\usepackage{amsmath}
\usepackage{xpatch}
\usepackage{amssymb}
\usepackage{esint}
\usepackage{mathtools}
\usepackage{amsthm}
\usepackage{nicefrac}
\usepackage{gensymb}
\usepackage{bigints}
\usepackage{mathtools}
\usepackage{blkarray, bigstrut}
\usepackage{physics}
\usepackage{calligra}
\usepackage{graphicx}
\usepackage{epstopdf}
\usepackage{dsfont}
\usepackage{array,ragged2e}
\usepackage{enumitem}
\usepackage{etoolbox}
\usepackage{babel}
\usepackage[nopar]{lipsum}
\usepackage{psfrag}
\usepackage[ruled,norelsize]{algorithm2e}
\usepackage{algorithm2e}
\usepackage{balance}
\usepackage{float}
\usepackage{verbatim}
\usepackage[caption = false]{subfig}
\SetKw{KwBy}{by}
\usepackage{algpseudocode}
\usepackage{makecell}
\usepackage{to-be-determined}
\usepackage{comment}
\usepackage[compact]{titlesec}
\usepackage{amssymb}

\usepackage{hyperref}

\titlespacing{\section}{0.5pt}{*0.5}{*0.5}
\titlespacing{\subsection}{0.5pt}{*0.5}{*0.5}
\titlespacing{\subsubsection}{0.5pt}{*0.5}{*0.5}

\makeatletter
\newcommand{\removelatexerror}{\let\@latex@error\@gobble}
\makeatother

\newtheoremstyle{plain}
  {\topsep}   
  {\topsep}   
  {\itshape}  
  {0pt}       
  {\bfseries} 
  {.}         
  {5pt plus 1pt minus 1pt} 
  {\thmname{#1}\thmnumber{ #2} \textnormal{(\thmnote{#3})}} 

\newtheorem{theorem}{Theorem}

\newtheorem{defn}{Definition}

\xpatchcmd{\proof}{\hskip\labelsep}{\hskip5\labelsep}{}{}  
\makeatletter
\xpatchcmd{\proof}{\@addpunct{.}}{\@addpunct{:}}{}{}
\makeatother

\renewcommand\[{\begin{equation}}
\renewcommand\]{\end{equation}} 
\usepackage{listings}
\usepackage{fancyvrb}
\usepackage{framed}

\usepackage{courier}

\definecolor{dkgreen}{rgb}{0,0.3,0}
\definecolor{gray}{rgb}{0.5,0.5,0.5}

\begin{document}

\title{Explanation-Guided Fair Federated Learning for Transparent 6G RAN Slicing}
\author{Swastika Roy,~\IEEEmembership{Student Member,~IEEE}, Hatim Chergui,~\IEEEmembership{Senior Member,~IEEE}\\and Christos Verikoukis,~\IEEEmembership{Senior Member,~IEEE}\\
\IEEEcompsocitemizethanks{\IEEEcompsocthanksitem S. Roy is with Iquadrat Informatica S.L. and Technical University of Catalonia (UPC), Barcelona, Spain, H. Chergui is with i2CAT Foundation, Barcelona, Spain and C. Verikoukis is with ATHENA/ISI and University of Patras, Greece. [e-mails: swastika.roy.roy@upc.edu, chergui@ieee.org, cveri@ceid.upatras.gr]}}

\maketitle
\thispagestyle{empty}

\begin{abstract}
Future zero-touch artificial intelligence (AI)-driven 6G network automation requires building trust in the AI black boxes via explainable artificial intelligence (XAI), where it is expected that AI faithfulness would be a quantifiable service-level agreement (SLA) metric along with telecommunications key performance indicators (KPIs). This entails exploiting the XAI outputs to generate transparent and unbiased deep neural networks (DNNs). Motivated by closed-loop (CL) automation and explanation-guided learning (EGL), we design an explanation-guided federated learning (EGFL) scheme to ensure trustworthy predictions by exploiting the model explanation emanating from XAI strategies during the training run time via Jensen-Shannon (JS) divergence. Specifically, we predict per-slice RAN dropped traffic probability to exemplify the proposed concept while respecting fairness goals formulated in terms of the \emph{recall} metric which is included as a constraint in the optimization task. Finally, the \emph{comprehensiveness score} is adopted to measure and validate the faithfulness of the explanations quantitatively.
Simulation results show that the proposed EGFL-JS scheme has achieved more than $50\%$ increase in terms of comprehensiveness compared to different baselines from the literature, especially the variant EGFL-KL that is based on the Kullback-Leibler Divergence. It has also improved the recall score with more than $25\%$ relatively to unconstrained-EGFL.

\end{abstract}

\vspace{0.5cm}
\section{Introduction}
\IEEEPARstart{F}{uture} 6G networks are expected  to be a \emph{machine-centric} technology, wherein all the corresponding \emph{smart things} would operate intelligently yet as smart black boxes \cite{XAI2} that lack transparency in their prediction or decision-making processes and could have adverse effects on the 6G network's operations. In this concern, XAI provides human interpretable methods to adequately explain the AI system and its predictions/decisions and gain the human's \emph{trust} in the loop, which is anticipated to be a quantifiable KPI \cite{confidence} in future 6G networks.
Nonetheless, most of the research on XAI for deep neural networks (DNNs) focuses mainly on generating explanations, e.g., in terms of feature importance or local model approximation in a \emph{post-hoc} way that does not bridge the trade-off that exists between explainability and AI model's performance \cite{tradeoff} nor guarantee the predefined 6G trust KPI requirement. In this regard, there is less contribution on aspects related to the exploitation of the explanations during training, i.e. in an \emph{in-hoc} fashion, which would orient the model towards yielding explainable outputs that fulfill a predefined trust level in AI-driven automated 6G networks, as well as tackling the explainability-performance trade-off \cite{tefl}. This approach is called \emph{explanation-guided learning} (EGL) \cite{salency,guide}, which focuses on addressing the accuracy of prediction and explanation jointly, and their mutual relation in run-time. It is particularly relevant in the context of 6G, where the complexity of the network and the diversity of the applications being served make it essential to guarantee the target trust in AI-driven zero-touch management. 
More specifically, recent advancements in network virtualization and software technology have given rise to network slicing \cite{ns1}, enabling the creation of virtual networks within a single physical network, with AI automating their management. In the realm of 6G, autonomous management of end-to-end network resources is favored over isolated slices due to efficiency and cost considerations. This approach is supported by the ETSI standardized zero-touch network and service management (ZSM) framework \cite{ZSM}. AI algorithms work with distributed datasets to fully leverage network slicing, necessitating decentralized learning, such as Federated Learning (FL). 
Like, Authors in \cite{7} highlight the importance of monitoring network slice performance for 5G's ZSM. They propose integrating FL into their ZSM framework to efficiently manage distributed network slices while ensuring data security, aligning with the broader goal of upholding trust in AI-driven zero-touch management.
In conclusion, ensuring transparency and trust in AI models is a technical necessity and a critical enabler for their integration into telecommunications networks, particularly in the 6G networks. These networks pose unique challenges, and addressing them is essential to ensure that AI is a reliable and valuable tool in the complex and highly regulated telecommunications environment.
So, this paper aims to present a distinct approach for improving federated learning models' interpretability by introducing a new closed-loop training procedure that natively optimizes explainability and fairness goals along with the main use-case task of traffic drop probability prediction for 6G network slices.
\vspace{0.5cm}

\section{Related Works}

In this section, we exhibit the state of the art of EGL to recognize interdisciplinary open research opportunities in this domain.
Indeed, the works on EGL techniques are still in the early stage and facing technical challenges in developing their frameworks \cite{EGL}. According to current studies, incorporating an additional explanation objective could potentially enhance the usefulness of DNNs in various application domains, including computer vision (CV) and natural language processing (NLP). However, its potential benefits for future 6G networks have not been thoroughly investigated yet. Despite this limited exploration, some researchers have started implementing XAI techniques to address telecommunication issues in this field. In particular, the authors of \cite{7,ex1} stress the importance of explainability and identify challenges associated with developing XAI methods in this domain. Additionally, the utilization of XAI for identifying the actual cause of Service Level Agreement (SLA) violations has been studied in \cite{ex2}.
Meanwhile, in \cite{XAI, neuro}, the authors mention the necessity of XAI in the zero-touch service management (ZSM) autonomous systems and present a neuro-symbolic XAI twin framework for zero-touch management of Internet of Everything (IoE) services in the wireless network. Moreover, \cite{FL_xai2} emphasizes the utilization of the Federated Learning (FL) approach with XAI in 5G/6G networks. In \cite{EGL}, two types of EGL termed global guidance and local guidance explanation are mentioned, where both types refine the model's overall decision-making process. Another line of work \cite{114} introduces the EGL framework that allows the user to correct incorrect instances by creating simple rules capturing the target model's prediction and retraining the model on that knowledge. While the authors of \cite{122} define a very generic explanation-guided learning loss called \emph{Right for the Right Reasons} loss (RRR). In addition, an approach to optimize model's prediction and explanation jointly by enforcing whole graph regularization and weak supervision on model explanation is presented in \cite{gens}.   
Additionally, a novel training approach for existing few-shot classification models is presented in \cite{Sun}, which uses explanation scores from intermediate feature maps. Specifically, it leverages an adapted layer-wise relevance propagation (LRP) method and a model-agnostic explanation-guided training strategy to improve model's generalization.
So it is apparent that none of the aforementioned works has considered EGL in telecommunication domains for prediction and decision-making. 

Many research papers explore decision-making processes and performance in diverse application domains. For instance, in a study by the authors \cite{pnr}, the focus is on employing mathematical optimization and iterative stochastic recovery algorithms (ISR) for network recovery when faced with large-scale failures and uncertain damage knowledge. While this approach is instrumental for network recovery, it may fall short in terms of interpretability, posing challenges for network operators who seek to comprehend and place trust in the rationale behind recovery decisions.Another paper of \cite{pfcc} centers on mobile blockchain networks and introduces a novel scheme utilizing federated semi-supervised learning. This scheme predicts link outage scale and efficiently restores connections between 5G network slices, resulting in improved consensus convergence. However, it does not directly address the aspects of explainability or transparency typically associated with AI.
On the other hand, AI-native zero-touch network slicing automation has been introduced to govern the distributed nature of datasets \cite{6} of various technological domains. Thus, a decentralized learning approach like federated learning (FL) \cite{7} has gained attraction for fulfilling the beyond 5G network requirement.
In this context, the authors of \cite{SFL} proposed a method for dynamic resource allocation in decentralized RAN slicing using statistical FL. This approach enables effective resource allocation, SLA enforcement, and reduced communication overhead. 
Similarly, the authors of \cite{dynamic} present a dynamic RAN slicing framework based on two-layer constrained Reinforcement Learning (RL) algorithm for vehicular networks, aiming to support multiple  Internet of vehicles (IoV) services and balance workloads among base stations (BSs). The objective is to minimize long-term system costs while satisfying coupled constraints and resource capacity constraints. While, the authors of \cite{sl} presents a novel SL(Split Learning) approach which parallelizes training by partitioning devices into clusters, reducing training latency. Additionally, a resource management algorithm is proposed to minimize CPSL's training latency, accounting for device variations and network dynamics in wireless networks.

However, the lack of a  clear understanding of the underlying mechanisms in the paper may raise concerns about the trust and transparency of the proposed model.
Furthermore, inspired by the work of \cite{farhad}, where the authors addressed a similar challenge, we build upon their foundation and draw inspiration from their eXplainable Reinforcement Learning (XRL) scheme. This approach seamlessly integrates eXplainable AI (XAI) with an explanation-guided strategy to enhance interpretability, with a focus on intrinsic interpretability to encourage Deep Reinforcement Learning (DRL) agents to learn optimal actions for specific network slice states.
In another facet of our work detailed in \cite{mine2}, we present a novel closed-loop eXplainable Federated Learning (EFL) approach. Our objective is to achieve transparent zero-touch service management of 6G network slices at RAN in a non-IID setup. In this study, we jointly consider explainability and sensitivity metrics as constraints during the Federated Learning (FL) training phase, specifically in the context of the traffic drop prediction task.
However, when comparing our current proposed work with the previously mentioned papers, we
recognize that although explainability has been introduced during the learning process for AI predictions or decision-making purposes, there remains to be a risk of bias or divergence from actual results in both approaches.


From this state-of-the-art (SoA) review, it turns out that incorporating FL and EGL strategies jointly can embrace into the network slicing of 6G RAN to address the absence of a trustworthy and interpretable solution in the prediction and decision-making process during learning. 
In this paper, we present the following contributions:
\begin{itemize}
\item To achieve transparent zero-touch service management of 6G network slices at RAN in a non-IID setup, a novel iterative explanation-guided federated learning approach is proposed, which leverages Jensen-Shannon divergence.

\item A constrained traffic drop prediction model and integrated-gradient based \emph{explainer} exchange attributions of the features and predictions in a closed loop way to predict per-slice RAN dropped traffic probability while respecting fairness constraints formulated in terms of recall. By demonstrating the traffic drop probability distribution and correlation heatmaps, we showcase the guided XAI's potential to support decision-making tasks in the telecom industry.
    
\item We frame the abovementioned fair sensitivity-aware constrained FL optimization problem under the \emph{proxy-Lagrangian} framework and solve it via a non-zero sum two-player game strategy \cite{{TwoPlayer}} and we show the comparison  with state-of-the-art solutions, especially the EGFL-KL. Further, to assess the effectiveness of the proposed one, we report the recall and comprehensiveness performance, which provides a quantitative measure of the sensitivity and EGFL explanation level.

\item More specifically, during the optimization process, we introduce a new federated EGL loss function that encompasses Jensen Shannon divergence to enhance model reliability and we also provide the corresponding results. We carefully scrutinize the normalized loss over the no. of FL rounds of the proposed EGFL model in the optimization process to support this phenomenon.
    
\end{itemize}


\section{Explanation Guided Learning Concept}
As mentioned before, the leading goal of EGL is to enhance both the model performance as well as interpretability by jointly optimizing model prediction as well as the explanation during the learning process. So, the main objective function of Explanation-Guided Learning is defined in the existing research work of \cite {EGL,50}, which is as follows: 

\begin{equation}
\begin{split}
    &\min\, \underbrace{\mathcal{L}_{Pred}\left(\mathcal{F}\left(\mathcal{X}\right), \mathcal{Y}\right)}_{prediction}+ 
    \underbrace{ \alpha\mathcal{L}_{Exp}\left({g}\left(\mathcal{F},\left(\mathcal{X,Y}\right)\right), \hat{M}\right) }_{explanation} \\&+ \underbrace{\beta\ohm\left({g}\left(\mathcal{F},\left(\mathcal{X,Y}\right)\right)\right)}_{regularization},
\end{split} 
\end{equation}
From equation (1), we can observe three terms in the objective function of EGL. Where the first term defines the prediction loss of the model, the second term deals with the supervision of model explanation considering some explicit knowledge, and the last term can introduce some properties for achieving the right reason. These three terms can be specified and enforced differently, pivoting on each explanation-guided learning method \cite{EGL}. However, it is noteworthy that we will adopt gradient based methods to explain the model's prediction in our proposed work. It generates attributions by computing the gradient of the model's output for its inputs, which indicates the input feature contributions to the model's final prediction. Theoretically, the gradient values of all unimportant features should be close to zero, which may not hold in some complex situations, giving biased results in the model decision that is unexpected in the actual scenario \cite{salency}. 
So, even if we introduce explainability during the learning process of any AI system for some prediction or decision-making purpose, there is still a chance we will get a final decision that is biased or different from actual results. Additionally, this may hinder telecommunication service providers/operators from making their system transparent, only considering XAI.
Moreover, other researchers have also addressed the challenge of achieving a trustworthy and interpretable solution in the learning process, similar to our approach. For instance, in \cite{traffic}, the authors tackled this issue and proposed an eXplainable Reinforcement Learning (XRL) scheme. Their method combines eXplainable AI (XAI) with an explanation-guided approach to enhance interpretability, focusing on intrinsic interpretability to encourage Deep Reinforcement Learning (DRL) agents to learn optimal actions for specific network slice states.

In our work detailed in \cite{mine}, inspired by \cite{traffic}, we presented a novel closed-loop eXplainable Federated Learning (EFL) approach. Our goal was transparent zero-touch service management of 6G network slices at RAN in a non-IID setup. In this study, we jointly considered explainability and sensitivity metrics as constraints during the Federated Learning (FL) training phase, specifically in the context of the traffic drop prediction task.

However, comparing our current proposed work with the mentioned papers, we acknowledge that despite introducing explainability during the learning process for AI predictions, there remains a risk of bias or divergence from actual results in both approaches. The issue arises from gradient-based eXplainable AI (XAI) methods, where expected near-zero gradients for unimportant features may not hold in complex scenarios, leading to biased model decisions.
Considering the above fact and inspired by EGL, we came out with a solution to overcome this situation. We aim to introduce EGL to our proposed framework of 6G RAN network slicing during learning, ensuring that the irrelevant gradient will close to zero without sacrificing the model performance.

\begin{table}[t]

\centering	
\newcolumntype{M}[1]{>{\centering\arraybackslash}m{#1}}

\caption{\textcolor{black}{Notations}}
{\color{black}\begin{tabular}{m{2cm}M{6cm}}
\hline
\hline
\multicolumn{1}{>{\centering\arraybackslash}M{2cm}}{Notation} & \multicolumn{1}{>{\centering\arraybackslash}M{6cm}}{Description}\\
\hline
$S_\mu(\cdot)$ & Logistic function with steepness $\mu$\\
$L$ & Number of local epochs\\
$T$ & Number of FL rounds\\
$\mathcal{D}_{k,n}$ & Dataset at CL $(k,n)$\\
$D_{k,n}$ & Dataset size\\
$\ell(\cdot)$ & Loss function\\
$\mathbf{W}_{k,n}^{(t)}$ &  Local weights of CL $(k,n)$ at round $t$\\
$\mathbf{x}_{k,n}^{(i)}$ & Input features\\
$\mathbf{y}_{k,n}^{(i)}$ & Actual Output \\
$\hat{y}_{k,n}^{(i)}$ & Predicted Output \\
$\gamma_{n}$ & Recall lower-bound for slice $n$\\
$\tilde{x}_{k,n}^{(i)}$ & Masked input features \\
$\tilde{p}_{k,n}^{(i)}$ & Masked predictions \\
$\mathcal{D}_{k,n}$ & Samples whose prediction fulfills the SLA \\
$\pi_{k,n}^{(i,j)}$ & Probability Distribution\\
$\rho_{k,n}$ & Recall Score\\
$\alpha_{k,n}^{(i,j)}$ & Weighted attribution of features \\
$\mathbf{JS}_{k,n}$ & Jensen Shannon Divergence \\
$\lambda_{(\cdot)}$ & Lagrange multipliers\\
$R_\lambda$ & Lagrange multiplier radius\\
$\mathcal{L}_{(\cdot)}$ & Lagrangian with respect to ($\cdot$)\\
\hline
\hline
\end{tabular}}
\label{notation}
\end{table}

\section{RAN Architecture and Datasets}
At first, we summarize the notations used throughout the paper in Table \ref{notation}.

As shown in Fig. \ref{network}, we consider a radio access network (RAN), which is composed of a set of $K$ the base station (BSs), wherein a set of $N$ parallel slices are deployed, driven by the concept of slicing-enabled mobile networks \cite{WirelessNS}. The deployment of these parallel slices within the RAN enables efficient resource allocation, differentiation of services, and flexibility to meet diverse user and application needs.

\begin{figure}[t]
\centering
 \includegraphics[width=0.60\textwidth]{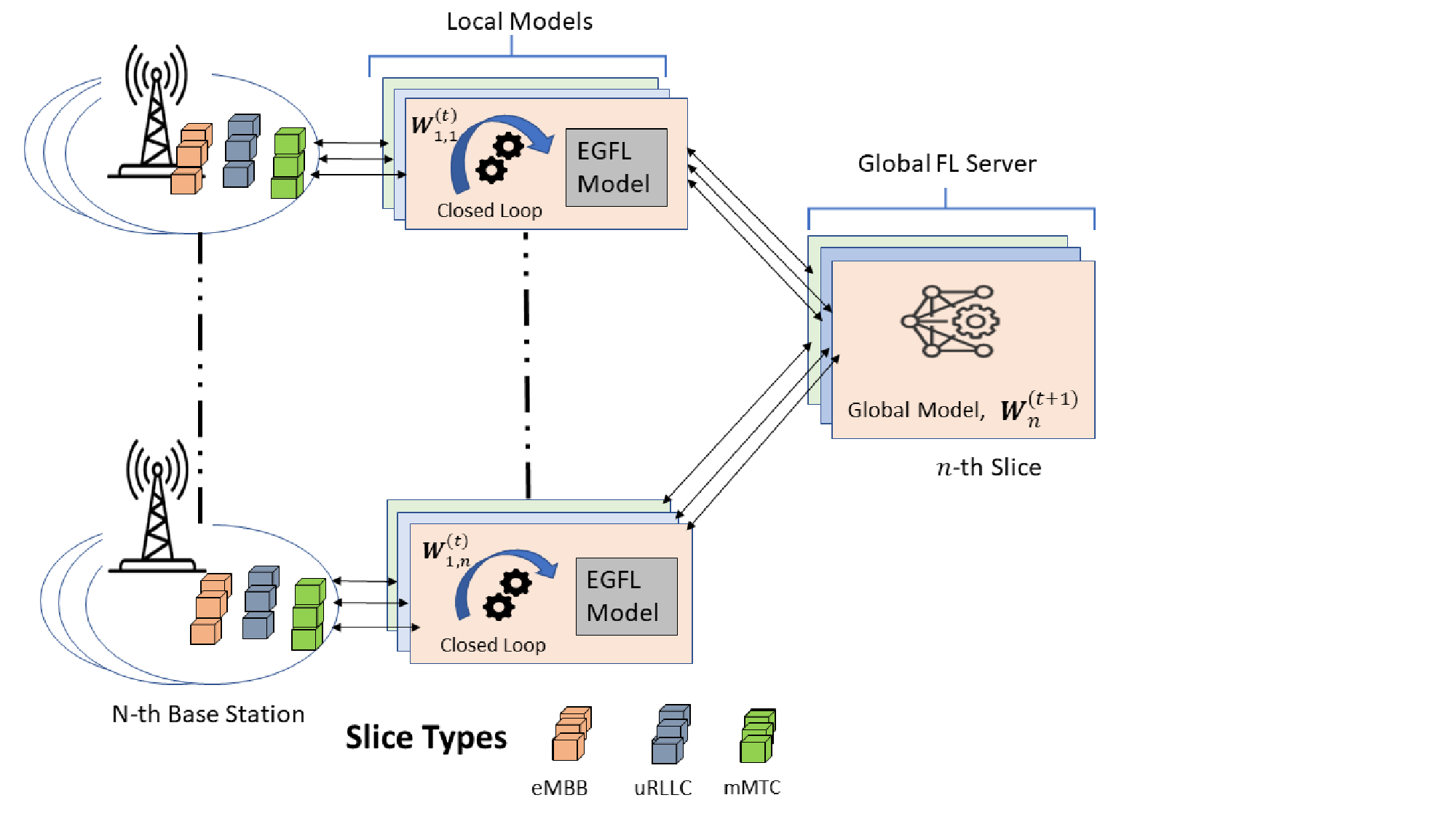}
\caption{Federated traffic drop probability prediction in 6G RAN NS.}
\vspace{-3mm}
\label{network}
\end{figure}
\begin{table}[t]
\label{Datasets_tab}
\centering	
\caption{Dataset Features and Output}
\color{black}\begin{tabular}{ll}
\hline
\hline 
Feature & Description\\
\hline
\texttt{Averege PRB} & Average Physical Resource Block\\ 
\texttt{Latency} &Average transmission latency\\
\texttt{Channel Quality} & SNR value expressing the wireless channel quality\\
\hline 
\hline 
& \\
\hline
\hline 
Output & Description\\
\hline
\texttt{Dropped Traffic} & Probability of dropped traffics (\%)\\
\hline
\hline
\label{Datasets_tab1}
\end{tabular}
\vspace{-2mm}
\label{feature}
\end{table}
\begin{figure*}[t!]
\centering
    \includegraphics[width=1.0\textwidth,scale=4, trim={0cm 0 0cm 0cm},clip]{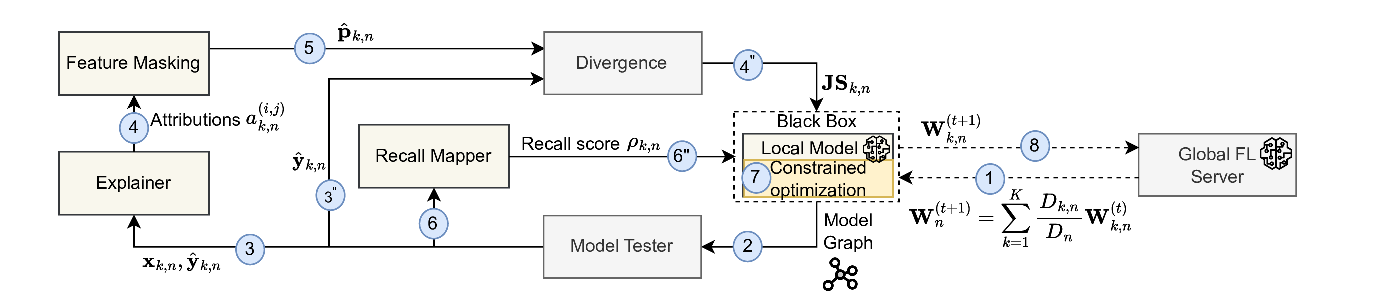}   
\caption{Explanation Guided FL building blocks}
\label{fig: egl}
\label{block}
\end{figure*}

Moreover, this solution's design closely adheres to the O-RAN framework \cite{ORAN}, a global initiative aimed at achieving increased openness in next-generation virtualized radio access networks (vRANs). Because, the  corresponding dataset employed in our study is derived from an O-RAN simulator, with simulations conducted within the framework detailed in the referenced paper \cite{traffic}.  To validate their proposed framework under realistic conditions, the authors considered the city of Milan, Italy, as their study scenario. They collect city-wide RAN deployment information, incorporating data from over 50 BSs from publicly available sources, and simulate realistic human mobility patterns based on prior work \cite{millan}. So, that O-RAN simulations involve the generation of diverse slice traffic profiles using heterogeneous Poisson distributions with time-varying intensities. Specifically, the profiles aim to accurately capture the unique behavior of each network slice, with a primary focus on analyzing dropped traffic that did not meet latency requirements due to incorrect Physical Resource Block (PRB) allocation decisions.
Here, each BS runs a local control closed-loop (CL) which collects monitoring data and performs traffic drop prediction. The local control CL system proactively manages network resources by gathering real-time data on network conditions and using predictive algorithms to estimate the likelihood of traffic drops or disruptions. Specifically, the collected data serves to build local datasets for slice $n \,(n=1,\ldots,N)$, i.e., $\mathcal{D}_{k,n}=\{\mathbf{x}_{k,n}^{(i)},y_{k,n}^{(i)}\}_{i=1}^{D_{k,n}}$, where $\mathbf{x}_{k,n}^{(i)}$ stands for the input features vector while $y_{k,n}^{(i)}$ represents the corresponding output. In this respect, Table \ref{feature} summarizes the features and the output of the local datasets, where each metric is averaged over a measurement period of $1$ second. These accumulated datasets are non-IID due to the different traffic profiles induced by the heterogeneous users' distribution and channel conditions.  Moreover, since the collected datasets are generally non-exhaustive to train accurate anomaly detection classifiers, the local CLs take part in a federated learning task wherein an E2E slice-level federation layer plays the role of a model aggregator. In more specifics, each slice within the BS represents a diverse use case of vertical applications and is characterized by input features such as PRBs, channel quality, and latency. By leveraging these input features, the FL optimization process can train a predictive model guided by XAI that learns to predict the occurrence of traffic drops for each slice within the BS. This enables proactive measures to be taken to mitigate potential disruptions and optimize the performance of the communication system.

\algblock{ParFor}{EndParFor}
\algnewcommand\algorithmicparfor{\textbf{parallel for}}
\algnewcommand\algorithmicpardo{\textbf{do}}
\algnewcommand\algorithmicendparfor{\textbf{end parallel for}}
\algrenewtext{ParFor}[1]{\algorithmicparfor\ #1\ \algorithmicpardo}
\algrenewtext{EndParFor}{\algorithmicendparfor}

\begin{algorithm}[t!]
\caption{Explanation Guided Federated Deep Learning}
\footnotesize
\SetAlgoLined
\KwIn{$K$, $m$, $\eta_{\lambda}$, $T$, $L$. \texttt{\# See Table II}}
\texttt{ \# Step 1: The algorithm starts with the initialization of model parameters $\mathbf{W}_n^{(0)}$ on the server, which is then broadcasted to the CLs} \\

Server initializes $\mathbf{W}_n^{(0)}$ and broadcasts it to the $\mathrm{CL}$s\\

\texttt{ \# The algorithm iterates for a specified number of rounds denoted by $T$.} \\
\For{$t=0,\ldots,T-1$}{
\texttt{\# Within each round, there is a parallelized loop over the CLs indexed by $k$.} \\
\texttt{\# For each CL, the algorithm initializes a local model $\mathbf{W}_{k,n,0}$ \\ 
and constructs a matrix $\mathbf{A}^{(0)}$.}

\textbf{parallel for} $k=1,\ldots,K$ \textbf{do}
Initialize $M=$ \texttt{num\_constraints} and $\mathbf{W}_{k,n,0}=\mathbf{W}_{n}^{(t)}$\\ 
Initialize $\mathbf{A}^{(0)}\in \mathbb{R}^{(M+1) \times (M+1)}$ with $\mathbf{A}_{m',m}^{(0)}=1/(M+1)$\\
\For {$l=0,\ldots,L-1$}{
\texttt{\# Step 2: Model graph from the local Model} \\
Receive the graph $\mathcal{M}_{k,n}$ from the local model \\
\texttt{\# Test the local model and calculate the attributions}\\
$a_{k,n}^{i,j}=$ \texttt{Int. Gradient} 
$\left(\mathcal{M}_{k,n}\left(\mathbf{W}_{k,n,l}, \mathbf{x}_{k,n}\right)\right)$\\ 

\texttt{\# Step 4: attributions are calculated using the Integrated Gradients method}

\texttt{\# Step 5: Mask the lowest dataset based on the attributions with zero padding,} $\tilde{\mathbf{x}}_{k,n}^{(i)}$ \\
\texttt{\# Step 3": Calculate the new model's predictions,} $\tilde{p}_{k,n}^{(i)}$\\
\texttt{\# Step 4": Calculate the Jensen-Shannon divergence,} $\mathbf{JS}(\hat{y}_{k,n}^{(i)}||\tilde{p}_{k,n}^{(i)})$\\
\texttt{\# Step 6-6": Calculate the recall metric}\\
$\rho_{k,n}= \pi^{+}\left(\mathcal{D}_{k,n}\left[\hat{y}_{k,n}^{(i)}=1\right]\right)$\\   
Let $\lambda^{(l)}$ be the top eigenvector of $\mathbf{A}^{(l)}$\\
\texttt{\# Step 7: Solve problem (\ref{OPT1}) via oracle optimization}\\
Let $\hat{\mathbf{W}}_{k,n,l}=\mathcal{O}_\delta\left(\mathcal{L}_{\mathbf{W}_{k,n,l}}(\cdot, \hat{\lambda}^{(l)})\right)$\\
Let $\Delta_{\lambda}^{(l)}$ be a gradient of $\mathcal{L}_{\lambda}(\hat{\mathbf{W}}_{k,n,l}, \lambda^{(l)})$ w.r.t. $\lambda$\\
\texttt{\# Exponentiated gradient ascent}\\
Update $\tilde{\mathbf{A}}^{(l+1)}=\mathbf{A}^{(l)}\odot\cdot\exp{\eta_{\lambda}\Delta_{\lambda}^{(l)}(\lambda^{(l)})}$\\
\texttt{\# Column-wise normalization}\\
$\mathbf{A}_{m}^{(l+1)}=\tilde{\mathbf{A}}_{m}^{(l+1)}/\norm{\mathbf{A}_{m}^{(l+1)}}_1,\,m=1,\ldots,M+1$\\
}
\texttt{\# Step 8: The final model for each CL is obtained by averaging  \\ the models from all local iterations: $\hat{\mathbf{W}}_{k,n}^{(t)}$.}\\

\Return{$\hat{\mathbf{W}}_{k,n}^{(t)}=\frac{1}{L^{\star}}\sum_{l=0}^{L-1}\hat{\mathbf{W}}_{k,n,l}$}\\
Each local CL $(k,n)$ sends $\hat{\mathbf{W}}_{k,n}^{(t)}$ to the server.\\ 
\textbf{end parallel for}\\
\texttt{\# The server aggregates these models to update its global model} \\
\Return{$\mathbf{W}_{n}^{(t+1)}=\sum_{k=1}^K \frac{D_{k,n}}{D_n}\hat{\mathbf{W}}_{k,n}^{(t)}$}\\
and broadcasts the value to all local CLs.}
\end{algorithm}

\section{Explanation-Guided FL for Transparent Traffic Drop Prediction}
In this section, we describe the different stages of the explanation-guided fair FL as summarized in Fig. \ref{block}.

\subsection{Closed-Loop Description}
In our proposed architecture, Fig. \ref{block} represents a block that is placed within the local models shown in Fig. \ref{network}. This block is designed to facilitate proposed explainable guided federated learning, which involves iterative local learning with runtime explanation in a closed-loop manner. We design a simple neural network FL model. For each local epoch, the Learner module feeds the posterior symbolic model graph to the Tester block which yields the test features and the corresponding predictions $\hat{y}_{k,n}^{(i)}$ to the Explainer. The latter first generates the features attributions using  integrated gradients XAI method.

The \emph{Feature Masking} block then selects, for each sample, the feature with lowest attribution and masks it. The corresponding masked inputs, $\tilde{x}_{k,n}^{(i)}, (i = 1,\ldots, D_{k,n}) $ are afterward used to calculate the masked predictions, $\tilde{p}_{k,n}^{(i)} , (i = 1,\ldots, D_{k,n})$ that is fed back to the Divergence module at step 5. At the same time, predictions $\hat{y}_{k,n}^{(i)}$ from the tester block are fed in the Divergence block as shown in the step 3". Divergence module is the responsible to calculate Jensen-Shannon divergence between model predictions of original and masked inputs, $\mathbf{JS}(\hat{y}_{k,n}^{(i)} ||\tilde{p}_{k,n}^{(i)})$.
The corresponding result, denoted as $\mathbf{JS}_{k,n}$ from this module is included with cross entropy loss as the objective function of the local model, shown at the stage 4". Similarly, the \emph{Recall Mapper} calculates the recall score $\rho_{k,n}$ based on the predicated and true positive values at stage 6 and 6" to include it in the local constrained optimization in step 7. 
Our focus lies in addressing imbalanced datasets within our traffic drop prediction problem, where accurately identifying positive classes assumes critical significance. The imbalance refers to the unequal distribution of classes in the dataset. Specifically, in traffic drop prediction, there are often fewer instances of traffic drop events compared to normal traffic. By incorporating recall constraints into our EGFL framework, the model is able to effectively capture positive instances, even those belonging to the minority class.
Indeed, the authors of \cite{oran-dataset} employed metrics such as recall and precision to gauge the performance of an AI service at the network edge in the O-RAN dataset.This reinforces our decision to prioritize "recall" as a constraint in our optimization problem. So, The rationale behind selecting recall as a crucial metric lies in its direct relevance to addressing fairness challenges, particularly in imbalanced datasets.
The recall, also known as sensitivity or true positive rate, represents the ability of a model to correctly identify all relevant instances of a particular class, often framed as the minority or underrepresented class in the dataset.
Our approach further involves framing a fair sensitivity-aware constrained FL optimization problem within the proxy-Lagrangian framework. This optimization problem is strategically addressed through a non-zero-sum two-player game strategy. In this game, each player represents a distinct aspect or objective of the overall strategy.
By introducing recall as a constraint in this game setting, one player is strategically motivated to make decisions that optimize the identification of positive instances.
So, this constraint prioritizes reducing false negatives, improving the overall detection of positive cases, and mitigating potential costs associated with missed instances.

Indeed, for each local CL $(k,n)$, the predicted traffic drop probability $\hat{y}_{k,n}^{(i)},\,(i=1,\ldots,D_{k,n})$, should minimize the main loss function with respect to the ground truth $y_{k,n}^{(i)}$ and Jensen-Shannon divergence score, while jointly respecting some long-term statistical constraints defined over its $D_{k,n}$ samples and corresponding to recall score. 
As shown in steps 1 and 7 of Fig. \ref{block}, the optimized local weights at round $t$, $\mathbf{W}_{k,n}^{(t)}$, are sent to the server which generates a global FL model for slice $n$ as,
\begin{equation}
\label{GlobalFL}
    \mathbf{W}_{n}^{(t+1)}=\sum_{k=1}^K \frac{D_{k,n}}{D_n}\mathbf{W}_{k,n}^{(t)},
\end{equation}
where $D_n=\sum_{k=1}^K D_{k,n}$ is the total data samples of all datasets related to slice $n$. The server then broadcasts the global model to all the $K$ CLs that use it to start the next round of iterative local optimization. Specifically, it leverages a two-player game strategy to jointly optimize over the objective  and original constraints as well as their smoothed surrogates and detailed in the sequel.

\subsection{Model Testing and Explanation}

As depicted in stage 2 of Fig. \ref{block}, upon the reception of the updated model graph, the Tester uses a batch drawn from the local dataset to reconstruct the test predictions $\hat{\mathbf{y}}_{k,n}^{(i)}$. All the graph, test dataset and the predictions are fed to the Explainer at stage 3. After that, at stage 4, Explainer generates the attributions
by leveraging the low-complexity Integrated Gradient (IG) scheme \cite{IG}, which is based on the gradient variation when sampling the neighborhood of a feature.
Attributions are a quantified impact of each single feature on the predicted output. Let $\mathbf{a}_{k,n}^{(i)}\in \mathbb{R}^{Q}$ denote the attribution vector of sample $i$, which can be generated by any attribution-based XAI method.

\subsection{Feature Masking}

To improve the faithfulness of the generated attributions, we introduce a novel approach to guide the following training. Specifically, we generate model's masked predictions, $\tilde{p}_{k,n}^{(i)}$, using the masked inputs, $\tilde{\mathbf{x}}_{k,n}^{(i)}$. Specifically, previously generated attributions score are used to select the lowest important features of each input sample and masking it with zero padding. Such masked dataset is employed only for evaluating the trustworthiness of the model, while the original, unmasked dataset is used during the optimization stage shown in Fig. \ref{block}. This ensures that the model gains knowledge from the entire, unaltered collection of features, avoiding any loss of essential data throughout the training process.

In this respect, the \emph{explanation-guided Mapper} at stage 5 of Fig. \ref{block} starts by selecting the lowest important features of each input sample based on their attributions which is collected from stage 4 and replace them with zero padding and after that new model output has generated. Finally this output are directly going to Divergence module.
\subsection{Divergence Module}
As mentioned above, We use this module to  calculate Jenson-Shannon divergence (JSD) score based on the predicted and masked predicted values, which we can define as $\mathbf{JS}_{k,n}\triangleq\mathbf{JS}\left(\hat{y}_{k,n}^{(i)} || \tilde{p}_{k,n}^{(i)}\right)$. It measures the similarity between two probability distributions while ensuring that the model's input or output data doesn't drastically change from a baseline.  It's derived from the Kullback-Leibler divergence (KL divergence) and provides a symmetric measure that quantifies the difference between two distributions.
The JSD between two probability distributions $\mathbf{P}$ and $\mathbf{Q}$ is defined as:
\begin{equation}
\text{JS}\left(P||Q\right)  = \frac{1}{2} \text{KL}\left(P|| M\right) + \frac{1}{2} \text{KL}\left(Q || M\right)
\end{equation}
Where, $\mathbf{KL}\left(P|| Q\right) $ represents the KL divergence from $\mathbf{Q}$ to $\mathbf{P}$, defined as:
\begin{equation}
\text{KL}(P || Q) = \sum_{i} P(i) \log \left(\frac{P(i)}{Q(i)}\right)
\end{equation}
and $\mathbf{M}$ represents the average distribution, defined as:
\begin{equation}
    M = \frac{1}{2}(P + Q)
\end{equation}
In our context, $\mathbf{P}$ represents as $\hat{y}_{k,n}^{(i)}$, the predicted distribution by the local model and $\mathbf{M}$ represents as $\tilde{p}_{k,n}^{(i)}$, the masked predicted distribution.
So, the idea behind including JSD along with the cross-entropy loss as our objective function is to train and encourage the models to generate samples similar to a target distribution considering both the original input and masked input. In this way, our model will learn to allocate low gradient values to irrelevant features in model predictions during its learning process. It will result in faithful gradients, giving unbiased and actual model decisions at the end of the learning.

\subsection{Explanation-Guided Fair Traffic Drop Prediction}
To ensure fairness, we invoke the \emph{recall} as a measure of the sensitivity of the FL local predictor. It is defined as the proportion of instances that truly belong to the positive class (drop traffic) and are correctly identified as such by the local predictor, which we denote $\rho_{k,n}$, i.e.,
\label{recall}
\begin{equation}
    \rho_{k,n}= \pi^{+}\left(\mathcal{D}_{k,n}\left[\hat{y}_{k,n}^{(i)}=1\right]\right)
\end{equation}
Where, $\pi^{+}(\mathcal{D}_{k,n})$ defines the proportion of $\mathcal{D}_{k,n}$ classified positive, and $\mathcal{D}_{k,n}[*]$ is the subset of $\mathcal{D}_{k,n}$ satisfying expression *.Note that, the condition for the positive class, indicating the drop traffics, is defined by $\hat{y}_{k,n}^{(i)} = 1$. 
 The subset $\mathcal{D}_{k,n}\left[\hat{y}_{k,n}^{(i)}=1\right]$ means that elements in the dataset $\mathcal{D}_{k,n}$ with a predefined condition are classified as belonging to the drop traffic category. So, the entire expression $\pi^{+}\left(\mathcal{D}_{k,n}\left[\hat{y}_{k,n}^{(i)}=1\right]\right)$ calculates the proportion of instances in $\mathcal{D}_{k,n}$ that are correctly classified as drop traffic, indicating the positive class for the recall calculation.

In order to trust the traffic drop anomaly detection/classification, a set of SLA is established between the slice tenant and the infrastructure provider, where a lower bound $\gamma_{n}$ is imposed to the recall score.
This translates into solving a constrained local classification problem for each local epoch in FL rounds $t\,(t=0,\ldots,T-1)$ i.e.,

\begin{subequations}
\label{OPT1}
\begin{equation}
    \min_{\mathbf{W}_{k,n}^{(t)}}\, \frac{1}{D_{k,n}}\sum_{i=1}^{D_{k,n}}\ell\left(y_{k,n}^{(i)}, \hat{y}_{k,n}^{(i)}\left(\mathbf{W}_{k,n}^{(t)},\mathbf{x}_{k,n}\right)\right) + \mathbf{JS}\left(\hat{y}_{k,n}^{(i)} || \tilde{p}_{k,n}^{(i)}\right)
\end{equation}

\begin{equation}
     \mathrm{s.t.}\hspace{5mm}\rho_{k,n} \geq \gamma_{n}\label{recall},
\end{equation}
\end{subequations}

Where, $\ell(.)$ is the cross-entropy loss function between the true labels ${y}_{k,n}^{(i)}$ and the model predictions $\hat{y}_{k,n}^{(i)}$ based on the current model weights $\mathbf{W}_{k,n}^{(t)}$ and input features $\mathbf{x}_{k,n}$ , while   $\mathbf{JS}(\hat{y}_{k,n}^{(i)}||\tilde{p}_{k,n}^{(i)})$ stands for Jenson-Shannon divergenceThe cross-entropy loss evaluates how well the model predictions match the true labels, and the Jenson-Shannon divergence quantifies the dissimilarity between the predicted distribution and masked distribution. Here, we choose not to set explicit weights for these components to ensure equal importance, balancing fidelity to true labels and consistency with the data distribution to meet SLA requirements effectively. This approach capitalizes on these components' inherent balance and complementary nature, addressing the specific needs of traffic anomaly detection.
The constraint enforces a minimum recall $\rho_{k,n}$ to meet the specified lower bound $\gamma_{n}$. This reflects the requirement from the SLA that the local predictor should achieve a certain level of sensitivity in identifying positive instances (traffic drop) to ensure the reliability of the anomaly detection/classification.
This optimization problem is solved by invoking the so-called \emph{proxy Lagrangian} framework \cite{Cotter}, since the recall is not a smooth constraint. This consists first on constructing two Lagrangians as follows:
\begin{subequations}
\label{ProxyLagrangian}
\begin{equation}
\begin{split}
    \mathcal{L}_{\mathbf{W}_{k,n}^{(t)}}=&\frac{1}{D_{k,n}}\sum_{i=1}^{D_{k,n}}\ell\left(y_{k,n}^{(i)}, \hat{y}_{k,n}^{(i)}\left(\mathbf{W}_{k,n}^{(t)},\mathbf{x}_{k,n}\right)\right)
    \\& + \mathbf{JS}\left( \hat{y}_{k,n}^{(i)} || \tilde{p}_{k,n}^{(i)}\right)
     +\lambda_1\Psi_1\left(\mathbf{W}_{k,n}^{(t)}\right),
\end{split}
\end{equation}
\begin{equation}
    \mathcal{L}_{\lambda}=\lambda_1\Phi_1\left(\mathbf{W}_{k,n}^{(t)}\right)
\end{equation}
\end{subequations}
where $\Phi_{1}$ and $\Psi_{1}$ represent the original constraints and their smooth surrogates, respectively. Here, $\mathcal{L}_{\mathbf{W}_{k,n}^{(t)}}$ minimizes the average loss over the dataset, the Jensen-Shannon divergence, and satisfies a constraint $\Psi_1$ using a Lagrange multiplier $\lambda_1$. While, $\mathcal{L}_{\lambda}$ uses a Lagrange multiplier $\lambda_1$ to enforce the original constraint  $\Phi_{1}$. The recall surrogate is given by,
\begin{equation}
    \Psi_1 = \frac{\sum_{i=1}^{D_{k,n}}{y}_{k,n}^{(i)} \times \min\Bigl\{\hat{y}_{k,n}^{(i)}, 1\Bigl\}}{\sum_{i=1}^{D_{k,n}} {y}_{k,n}^{(i)}} - \gamma_n
\end{equation}
It represents the proportion of actual positive instances that were correctly predicted, considering the lower bound of the recall score. Specifically, $\min({\hat{y}_{k,n}^{(i)},1})$ function takes the minimum value between the predicted value $\hat{y}_{k,n}^{(i)}$ and 1. The purpose is to ensure that the predicted value does not exceed 1, which is relevant when dealing with probability scores.The numerator ${\sum_{i=1}^{D_{k,n}}{y}_{k,n}^{(i)} \times \min\Bigl\{\hat{y}_{k,n}^{(i)}, 1\Bigl\}}$  represents the sum of the minimum values between the predicted scores and 1 for instances where the true label is positive (${y}_{k,n}^{(i)} = 1$). This is essentially the count of instances where the model correctly predicted positive instances (true positives), considering the lower bound of the recall score.In summary, the equation is a surrogate for recall, emphasizing the correct prediction of actual positive instances while considering the lower bound of the predicted values, especially when dealing with probability scores. The multiplication in the numerator ensures that only cases with true positive labels are considered in the summation, contributing to the recall calculation.

So, this optimization task turns out to be a non-zero-sum two-player game in which the $\mathbf{W}_{k,n}^{(t)}$-player aims at minimizing $\mathcal{L}_{\mathbf{W}_{k,n}^{(t)}}$, while the $\lambda$-player wishes to maximize $\mathcal{L}_{\lambda}$ \cite[Lemma 8]{TwoPlayer}. While optimizing the first Lagrangian w.r.t. $\mathbf{W}_{k,n}$ requires differentiating the constraint functions $\Psi_1(\mathbf{W}_{k,n}^{(t)})$, to differentiate the second Lagrangian w.r.t. $\lambda$ we only need to evaluate $\Phi_1\left(\mathbf{W}_{k,n}^{(t)}\right)$. Hence, a surrogate is only necessary for the $\mathbf{W}_{k,n}$-player; the $\lambda$-player can continue using the original constraint functions. The local optimization task can be written as,
\begin{subequations}
\label{ProxyLagrangian}
\begin{equation}
\begin{split}
    \min_{\mathbf{W}_{k,n}\in \Delta} \,\,\,\,\max_{\lambda,\, \norm{\lambda}\leq R_\lambda}\,\,\mathcal{L}_{\mathbf{W}_{k,n}^{(t)}}
\end{split}
\end{equation}
\begin{equation}
        \max_{\lambda,\, \norm{\lambda}\leq R_\lambda}\,\,\,\,\min_{\mathbf{W}_{k,n}\in \Delta} \mathcal{L}_{\lambda},
\end{equation}
\end{subequations}
where thanks to Lagrange multipliers, the $\lambda$-player chooses how much to weigh the proxy constraint functions, but does so in such a way as to satisfy the original constraints, and ends up reaching a nearly-optimal nearly-feasible solution \cite{Gordon}. These steps are all summarized in Algorithm 1.

\begin{figure*}[thb]
    \centering
    \subfloat[ \centering Analysis of convergence with two approaches]{
    \label{loss_a}
          \includegraphics[width=0.40\textwidth]{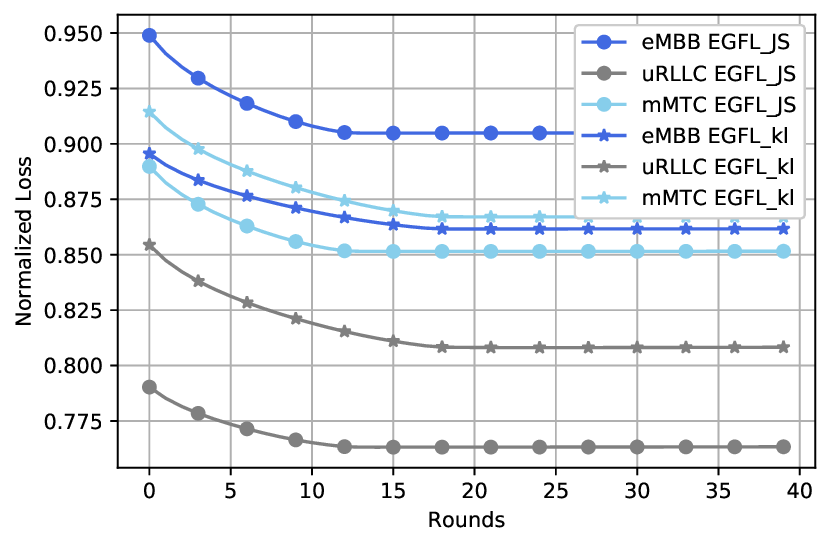}\hspace{1.2cm}}
    \qquad 
    \subfloat[ \centering Analysis of convergence with additional approaches for eMBB slice]{
    \label{loss_eMBB}
          \includegraphics[width=0.40\textwidth]{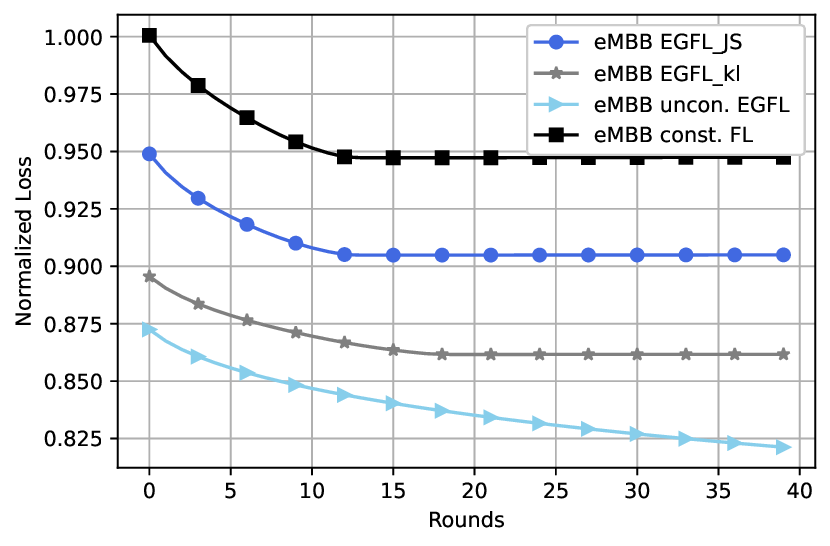}}
     \qquad 
    \subfloat[ \centering Analysis of convergence with additional approaches for uRLLC slice]{
    \label{loss_b}
          \includegraphics[width=0.40\textwidth]{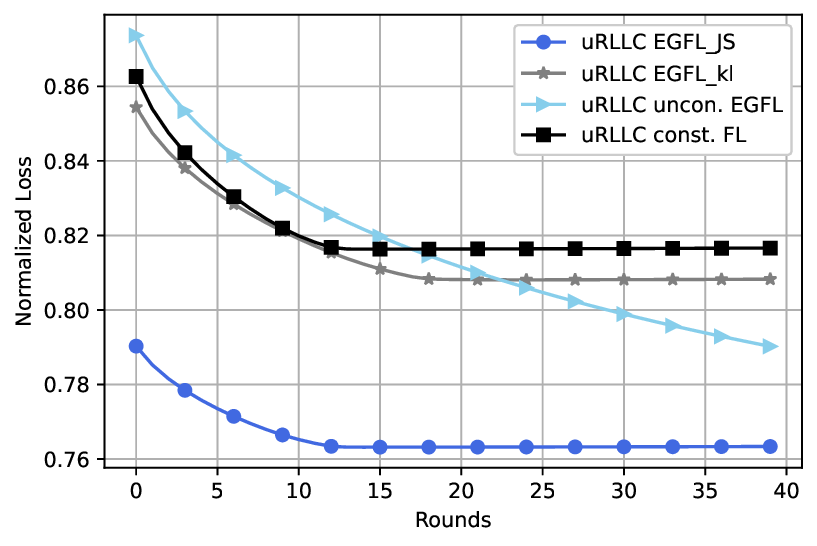}}
     \qquad 
    \subfloat[ \centering Analysis of convergence with additional approaches for mMTC slice]{
    \label{loss_mmtc}
          \includegraphics[width=0.40\textwidth]{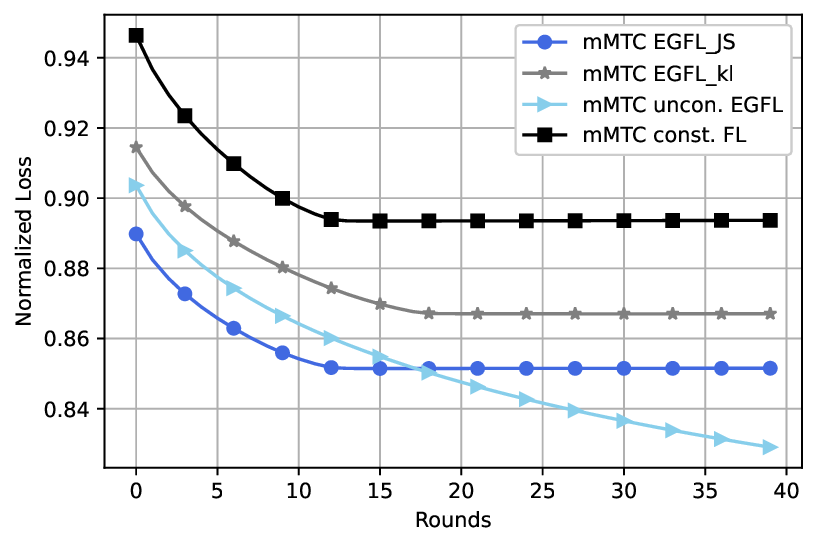}}
    
\caption{Analysis of FL training loss vs FL rounds of EGFL with lower bound of recall score, $\gamma = [0.82, 0.85, 0.84]$ }
\label{loss}
\end{figure*}

\section{EGFL Convergence Analysis}

Within this section, we examine the probability of convergence for the explanation-guided federated learning method.
This analysis draws inspiration from our previous work \cite{offline}, which focuses on resource provisioning in SDN/NFV-based 5G networks. The prior work introduces a traffic predictor, employs multi-slice DNNs for resource estimation under SLAs, and integrates non-convex constraints through a game strategy. Hyperparameters balance overprovisioning and isolation, and the study provides a reliability-based analysis of convergence probability influenced by violation rates.
To accomplish this, we establish a closed-form formula that determines the minimum probability of convergence, considering various factors such as the Lagrange multiplier radius, the optimization oracle error, and the violation rate, while incorporating the Jensen-Shannon divergence.
 In practical implementation, we apply the specified objective function along with its associated constraints and proxy constraints using Google's constrained optimization package \cite{constrained}. This package employs two distinct approaches for optimizing the Lagrangians: a Bayesian optimization oracle for $\mathcal{L}_{W,n}$ and projected gradient ascent for $\mathcal{L}_{\lambda}$.The oracle, which is integral to the optimization process, can be defined as follows:
\begin{defn}[Approximate Bayesian Optimization Oracle]
An oracle $\mathcal{O}{\delta}$ for $\delta$-approximate Bayesian optimization refers to a procedure that takes a loss function/Lagrangian $\mathcal{L}$ as input and outputs quasi-optimal weights $\mathbf{W}{k,n,l}$, satisfying,
\begin{equation}
\mathcal{L}\left(\mathcal{O}_{\delta}\left(\mathcal{L}\right)\right)\leq \inf_{\mathbf{W}_{k,n,l}^{\star}} \mathcal{L} \left(\mathbf{W}_{k,n,l}^{\star}\right)+\delta. 
\end{equation}
\end{defn}
The key theorem can therefore be formulated as,
\begin{theorem}[EGFL Convergence Analysis]\label{Thm1}
We consider that EGFL follows a geometric failure model and fails to meet the recall constraint with an average violation rate $0<\nu<1$. It is also assumed that an oracle $\mathcal{O}_{\delta}$ with error---$\delta$ excluding the JS divergence---optimizes $\mathcal{L}_{\mathbf{W}{k,n}^{(t)}}$, while $R_\lambda$ and $B_{k,n}$ represent the Lagrange multipliers radius and the maximum subgradient norm, respectively. In such a scenario, the convergence probability of EGFL can be expressed as,

\begin{equation}
\begin{aligned}
    \mathbf{Pr} \bigg[\frac{1}{T}\sum_{t=1}^{T}&\Big(\mathcal{L}\left(\mathbf{W}_{n}^{(t)},\lambda^{\star}\right)\\ &- \inf_{\mathbf{W}_n^{\star}} \mathcal{L}\left(\mathbf{W}_n^{\star},\lambda^{(t)}\right)\Big) < \epsilon \bigg]\geq \Delta\left(\nu, \epsilon\right),\label{conv}
    \end{aligned}
\end{equation}
where
\begin{equation}
    \Delta\left(\nu, \epsilon\right) = 1 - \frac{\nu}{1+(\nu-1)\exp{\frac{-(D_n\epsilon)^2}{2(2R_\lambda \sum_{k=1}^{K}D_{k,n}B_{k,n}+D_n\alpha)^2}}}.\label{Q}
\end{equation}
and $\alpha = \delta+\ln(1-V^2/4)$ while $V$ stands for the total variation distance between $\hat{y}_{k,n}^{(i)}$ and $\tilde{p}_{k,n}^{(i)}$.
\end{theorem}

\begin{proof}
First, relying on the subgradient inequality, we have at round $t$,
\begin{equation}
    \mathcal{E}^{(t)}=\mathcal{L}\left(\mathbf{W}_n^{(t)},\lambda^{*}\right) -  \mathcal{L}\left(\mathbf{W}_n^{(t)},\lambda^{(t)}\right)\leq \langle\nabla\mathcal{L}^{(t)},\lambda^{*}-\lambda^{(t)}\rangle.
\end{equation}
Then, by decomposing the
By means of Cauchy-Schwarz inequality, we get,
\begin{equation}
    \mathcal{E}^{(t)} \leq \norm{\nabla\mathcal{L}\left(\mathbf{W}_{n}^{(t)},\lambda^{(t)}\right)}\norm{\lambda^{\star}-\lambda^{(t)}}. \label{Holder}
\end{equation}
By invoking the federated learning aggregation (\ref{GlobalFL}), we can write
\begin{equation}
    \nabla\mathcal{L}\left(\mathbf{W}_{n}^{(t)},\lambda^{(t)}\right)=\sum_{k=1}^K \frac{D_{k,n}}{D_n}\nabla\mathcal{L}\left(\mathbf{W}_{k,n}^{(t)},\lambda^{(t)}\right).
    \label{GradFed}
\end{equation}
Therefore, from (\ref{Holder}) and (\ref{GradFed}) and by invoking the triangle inequality we have,
\begin{equation}
\begin{aligned}
   \mathcal{E}^{(t)}&\leq \sum_{k=1}^K \frac{D_{k,n}}{D_n}\norm{\nabla\mathcal{L}\left(\mathbf{W}_{k,n}^{(t)},\lambda^{(t)}\right)}\norm{\lambda^{\star}-\lambda^{(t)}}\\
    &\leq2R_\lambda \sum_{k=1}^K \frac{D_{k,n}}{D_n} B_{k,n}. \label{triangle}
\end{aligned}
\end{equation}
At this level, let
\begin{equation}
    \mathcal{U}^{(t)}=\mathcal{L}\left(\mathbf{W}_n^{(t)},\lambda^{\star}\right) -  \inf_{\mathbf{W}_n^{\star}} \mathcal{L}\left(\mathbf{W}_n^{\star},\lambda^{(t)}\right). \label{U} 
\end{equation}
On the other hand, from \cite{bounds_js, js}, we have
\begin{equation}
    \mathbf{JS}_{k,n} > - \ln(1 - \frac{V^2}{4}), \label{ineq_js}
\end{equation}

Combining (\ref{triangle}), (\ref{U}) and (\ref{ineq_js}) with Definition 1 yields, 
\begin{equation}
  \mathcal{U}^{(t)}\leq 2R_\lambda \sum_{k=1}^K \frac{D_{k,n}}{D_n} B_{k,n} + \delta + \ln(1 - \frac{V^2}{4})= C.
\end{equation}
By means of Hoeffding-Azuma's inequality \cite{hoeffding}, we have,
\begin{equation}
    \Pr[\frac{1}{T}\sum_{t=1}^{T}\mathcal{U}^{(t)}< \epsilon\mid T=\tau]\geq 1 - \exp{-\frac{\tau\varepsilon^2}{2C^2}},\label{Azuma},
\end{equation}
where we consider that the federated learning model is fair, i.e., respecting the recall constraint up to and including time $T=\tau$. Therefore, by recalling the geometric failure probability mass function $P_{\tau}$ given by,

\begin{equation}
    P_{\tau} = \nu\left(1 - \nu \right)^{\tau},
\end{equation}
and combining it with (\ref{Azuma}), yields
\begin{equation}
\Pr[\frac{1}{T}\sum_{t=1}^{T}\mathcal{U}^{(t)}< \epsilon]\geq\sum_{\tau=0}^{+\infty}\nu\left(1 - \nu \right)^{\tau}\times\left(1-\exp{-\frac{\tau\epsilon^2}{2C^2}}\right).
\end{equation}
Finally, by noticing that $\nu < 1$ and using the infinite geometric series closed-form, we obtain the target result as in (\ref{conv}) and (\ref{Q}).
\end{proof}

\begin{table}[t]
\centering	
\caption{Settings}
\begin{tabular}{ccc}
\hline 
\hline
Parameter & Description & Value\\
\hline
$N$ & \# Slices & $3$\\
$K$ & \# BSs & $50$\\ 
$D_{k,n}$ & Local dataset size & $1500$ samples\\ 
$T$ & \# Max FL rounds & $40$\\  
$U$ & \# Total users (All BSs)& $15000$\\
$L$ & \# Local epochs & $40$\\ 
$R_{\lambda}$ & Lagrange multiplier radius &  Constrained: $10^{-5}$\\
$\eta_{\lambda}$ & Learning rate & $0.12$ \\ 
\hline
\hline
\label{FLsettings}
\end{tabular}
\label{setting}
\end{table}

\section{Results}
This section scrutinizes the proposed Closed loop EGFL framework.
The dataset used for analysing this framework has already mentioned in the section IV. We use feature attributions generated by the XAI to build a fair sensitivity-aware, explanation-guided, constrained traffic drop prediction model.
We also present the results corresponding to the fairness recall constraint and discuss its impact in detail. Moreover, we show FL convergence results and, finally, we analyze the comprehensiveness score of all slices to show the faithfulness of our proposed model. Notably, we invoke \texttt{DeepExplain} framework, which includes state-of-the-art gradient and perturbation-based attribution methods \cite{Deep} to implement the Explainer for generating attribution scores of the input features. We integrate those scores with our proposed fair EGFL framework in a closed-loop iterative way.

\subsection{Parameter Settings and Baseline}
Three primary slices eMBB, uRLLC and mMTC are considered to analyze the proposed EGFL. Here, the datasets are collected from the BSs and the overall summary of those datasets are presented in Table \ref{setting}. We use vector $\alpha$ for the lower bound of recall score corresponding to the different slices. As a baseline, we adopt a vanilla FL \cite{Van_FL} with post-hoc integrated gradient explanation, that is, a posterior explanation performed upon the end of the FL training.


\subsection{Result Analysis}
In this scenario, resources are allocated to slices according to their traffic patterns and radio conditions.In this section, we thoroughly compare our proposed EGFL-JS scheme to other compatible state-of-the-art baselines from various perspectives, such as convergence, fairness, and faithfulness analysis, to show its superior efficacy. 
\begin{figure}[t]
\centering
     \includegraphics[width=0.41\textwidth]{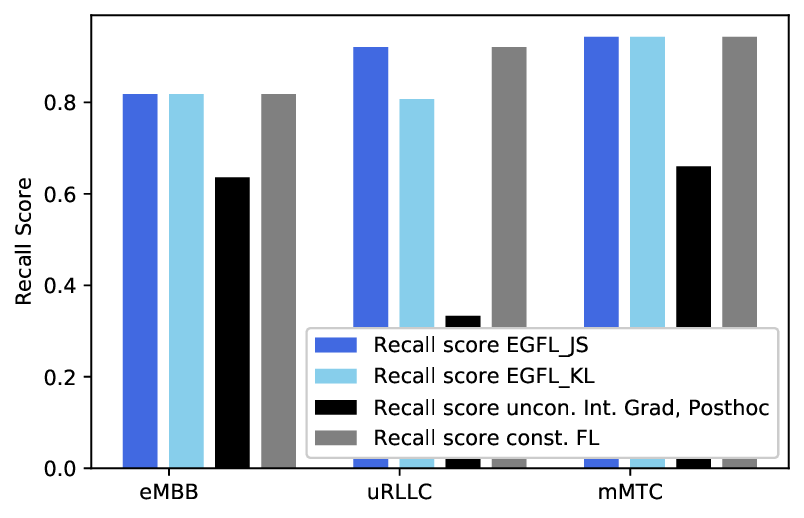}
\caption{Analysis of recall score with lower bound $\gamma = [0.82, 0.85, 0.84]$ }
\label{recall}
\end{figure}

\begin{figure}[t]
\centering
     \includegraphics[width=0.47\textwidth]{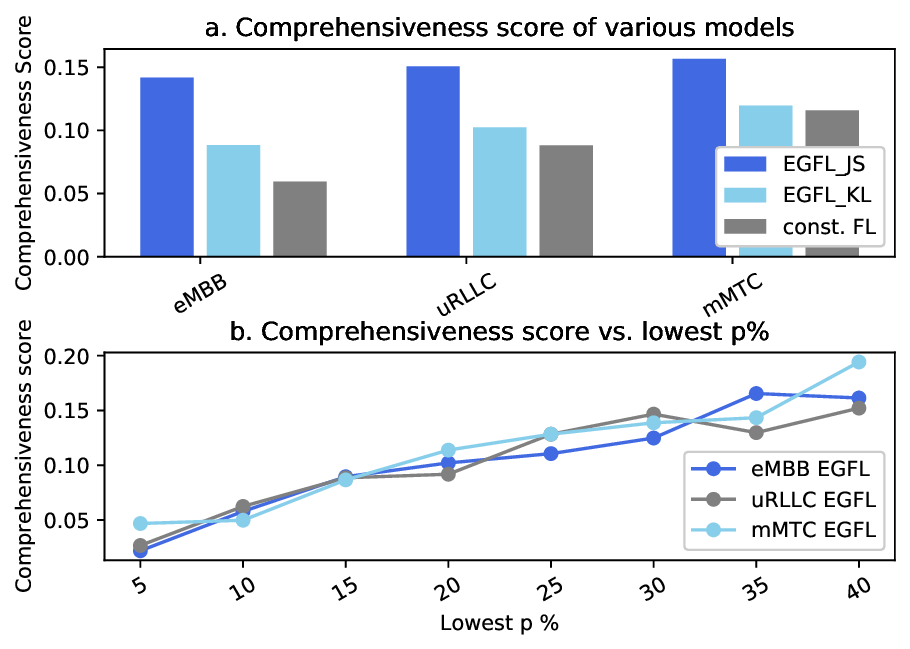}
\caption{Analysis of comprehensiveness score with lower bound of recall score, $\gamma = [0.82, 0.85, 0.84]$}
\label{comp}
\end{figure}

\begin{figure*}[thb]
    \centering
    \subfloat[ \centering Distribution plot for uRLLC network slice]{
    \label{att_dist}
        \includegraphics[width=0.46\textwidth,height=100pt]{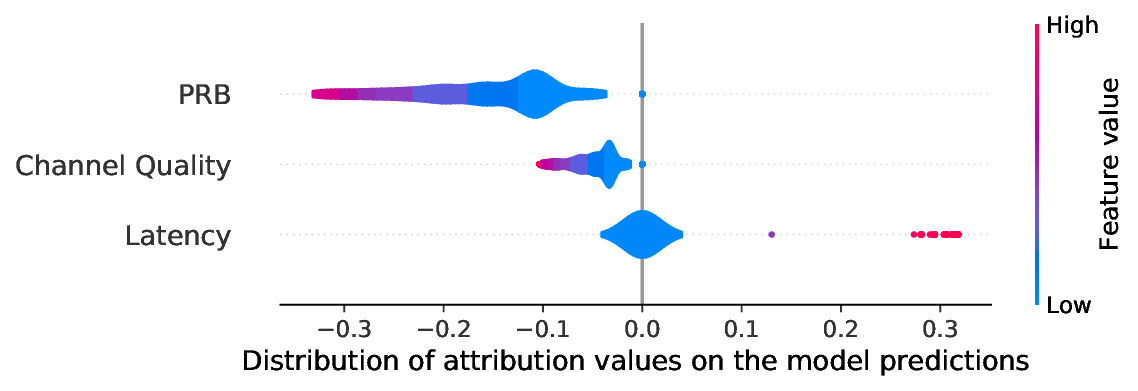}\hspace{1.2cm}}
    \qquad 
    \subfloat[ \centering Attribution-based feature impact on the model output for uRLLC network slice]{
    \label{att_avg}
          \includegraphics[width=0.38\textwidth]{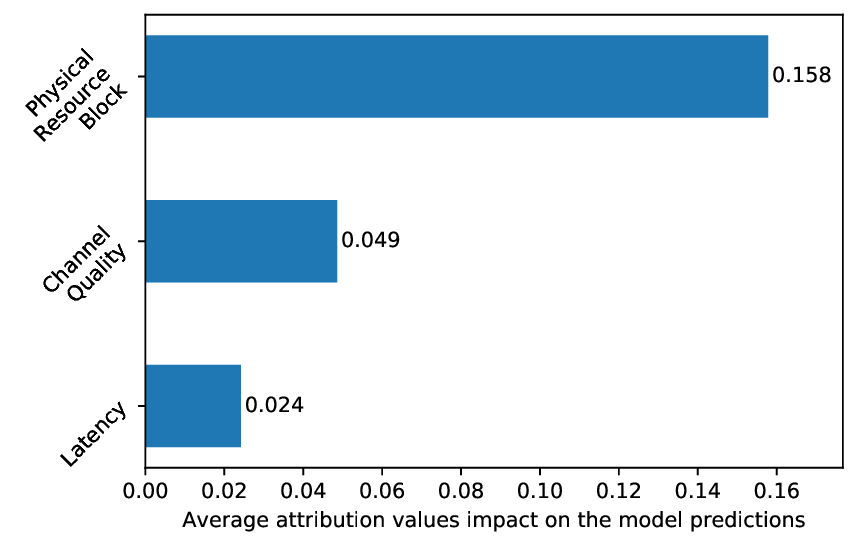}}
\caption{Root cause analysis of EGFL-JS predictions in a telecommunication network.}
\end{figure*}

\begin{figure}[htb]
\centering
\label{heatmap}
     \includegraphics[width=0.41\textwidth]{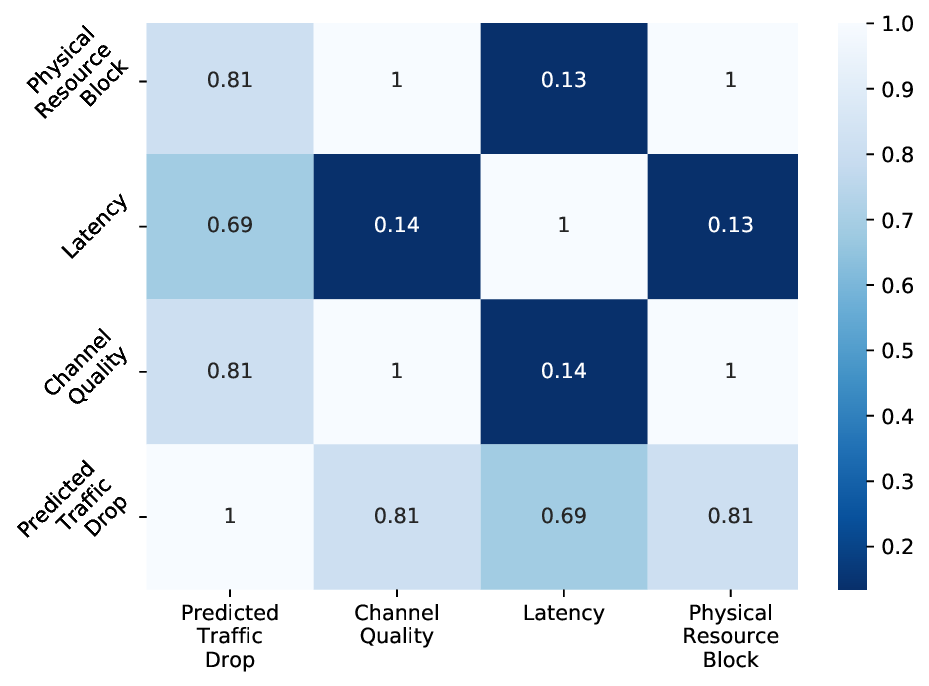}
\caption{Correlation Heatmap for uRLLC network slice}
\label{heatmap}
\end{figure}

\begin{itemize}
\item \textbf{Convergence:}    
First, we have plotted FL normalized training loss vs. FL rounds of EGFL-JS, and shown the comparison with another compatible strategy termed EGFL-KL. The key distinction between the two techniques is the divergence measure used to guide their respective loss functions. The JSD guides the cross-entropy loss function in our proposed solution, whereas the EGFL-KL model uses the Kullback-Leibler Divergence (KL Divergence) to guide its loss function. It is known that JSD, a symmetrized variant of KL Divergence, is less sensitive to severe discrepancies across distributions, which may help speed up convergence. 
We visually depict both models' convergence behavior across several slices in Fig. \ref{loss_a}.
In analyzing this figure, clear distinctions emerge in the normalized loss across the three slices (eMBB, mMTC, and uRLLC). These variations can be attributed to unique factors such as diverse traffic patterns, data fluctuations, slice-specific demands, and the models' sensitivity. Specifically, slices like eMBB and mMTC, characterized by more dynamic traffic, may undergo a higher normalized loss during the convergence process. In contrast, uRLLC, associated with potentially more stable traffic conditions, demonstrates a lower normalized loss. This indicates that the observed differences are intricately tied to each slice's inherent characteristics and requirements.

Notably, compared to the EGFL-KL model, our suggested model consistently shows a higher rate of convergence. This finding is consistent with what we predicted since JSD is naturally less susceptible to being impacted by severe distribution discrepancies, which results in a smoother optimization landscape and, as a result, faster convergence. Additionally, we broaden the scope of our research beyond the EGFL-KL model by including two additional methods: unconstrained EGFL (without recall constrained) and recall-constrained FL (without Jensen-Shannon Divergence). 
The results from Fig. \ref{loss_eMBB}, Fig. \ref{loss_b}, and Fig. \ref{loss_mmtc} indicate that the proposed constrained EGFL strategy converges faster for all slices compared to unconstrained EGFL. Additionally, EGFL exhibits lower training losses compared to constrained FL. The absence of a recall constraint in unconstrained EGFL leads to a continuous decrease in normalized loss without convergence, as the model is not explicitly optimized for accurate positive instance identification, particularly in unbalanced datasets. In contrast, incorporating recall constraints during training with the proposed approach, such as EGFL-KL, EGFL-JS, and const. FL, facilitates fast convergence across all slices. So, the empirical results of Fig. \ref{loss} support the improved convergence characteristics of our proposed model, especially with the inclusion of JSD in the loss function. These findings underscore the efficiency and effectiveness of our approach in various real-world scenarios.

\item \textbf{Fairness Analysis:} 
As illustrated in Fig. \ref{recall}, the recall metric has been chosen to analyze our proposed model's fairness. From that figure, we can observe that the recall score of the proposed EGFL-JS for all slices is closer to the target threshold $\lambda$ (i.e., around 0.80 $\%$), which would be an acceptable value for slices' tenants.
In summary, the recall scores of EGFL-JS, EGFL-KL and constrained FL for all slices are almost identical and higher than the unconstrained EGFL. In comparison, the training loss value of EGFL is lower than other approaches. The model with divergence included is better at generalizing to new, unseen data. However, having the same recall score indicates that the model's ability to identify positive samples has not been affected by the inclusion of divergence, which is a favorable sign. This can be fruitful in cases where the model will be deployed in real-world applications, as it suggests that the model is more likely to perform well on new data that it has not seen during training.
It gives us an approximate idea of our model's reliability and trustworthiness.

\item \textbf{Faithfulness of Explanation:} Comprehensiveness evaluates if all features needed to make a prediction are selected \cite{EGL}. The model comprehensiveness is calculated as:
Comprehensiveness = $\hat{y}_{k,n}^{(i)} - \tilde{p}_{k,n}^{(i)}$.
Here, the higher score implies that the explanation provided is complete, comprehensive, and so helpful to the user. In Fig. \ref{comp} we have plotted the comprehensiveness score of all the slices, where we notice that the proposed EGFL has a higher score than the EGFL-KL and baseline one. Moreover, in Fig. \ref{comp}b, we monitor the effect of the changing value of the $p\%$ features removal on the comprehensiveness score, considering the proposed model for all slices.
Specifically, during the calculation of comprehensiveness score based on the removal of $p\%$ of features with lowest attributions, we have created a new masked input, denoted as $\tilde{\mathbf{x}}_{k,n}^{(i)}$ for each $p\%$ value and also have generated the corresponding masked predictions, $\tilde{p}_{k,n}^{(i)}$.
Here, the score increases versus the lowest $p\%$, which means that the masked features are decisive and the corresponding explanation in terms of attribution is comprehensive.

\item \textbf{Root cause analysis in wireless networks and Interpretation:} 
Feature attributions generated by XAI methods reflect the contribution of the input features to the output prediction. From a telecommunication point of view, this can be leveraged to perform root cause analysis of different network anomalies and KPIs degradation. In this respect, we investigate the factors that impact slice-level packet drop by plotting, for instance, the distribution of the uRLLC slice feature attributions in Fig. \ref{att_dist} and their average in Fig. \ref{att_avg}. The first subfigure demonstrates that EGFL traffic drop predictions are negatively influenced by PRB and the channel quality, since the attributions are distributed and concentrated towards negative values, implying that higher values of these features reduce the likelihood of traffic drops. It is noted that PRB has a slightly higher impact compared to channel quality. In this regard, Fig. \ref{att_avg} further supports these observations by confirming that in average PRB has the most influence on the model's output, reaffirming the importance of PRB allocation in reducing traffic drops. This observation aligns with the notion that higher PRB allocation implies dedicating more resources to the uRLLC slice, which in turn relaxes the transmission queues and reduces the packet drop. Lastly, the distribution and average plots reveal that latency has a narrower and more focused distribution near zero. This suggests that latency may have lesser significance in predicting traffic drops compared to other features.
Furthermore, the correlation heatmap of Fig. \ref{heatmap} indicates the strength and direction of the linear relationship between the predicted traffic drops and each input feature, namely PRB, channel quality, and latency with a correlation of 81 $\%$,  81 $\%$ and  69 $\%$ respectively., which corroborates the previous analysis. Note that to plot correlation matrix heatmap, we consider one matrix, $\mathbf{R}_{k,n}$ = [$\mathbf{a}_{k,n},\hat{\mathbf{y}}_{k,n}$], where, $\mathbf{a}_{k,n}$ denotes the attribution score of features with dimensions $D_{k,n} \times Q$ and $\hat{\mathbf{y}}_{k,n}$ is the predicted output variable with dimensions $D_{k,n} \times 1$.
A correlation of 81 $\%$ between the predicted traffic drops and PRB as well as channel quality suggests a strong positive relationship. This means that higher values of PRB and channel quality are associated with a higher likelihood of lower traffic drops in the uRLLC slice.
On the other hand, the correlation coefficient of 69 $\%$ between the predicted traffic drops and latency indicates a moderate positive relationship. This means that higher values of latency are generally associated with a higher likelihood of traffic drops in the uRLLC slice.
Based on these observations, it can be inferred that when dealing with real-time traffic in the telecom industry, maintaining an appropriate allocation of PRB and monitoring channel quality in the network slice becomes crucial for ensuring better services to users.

\end{itemize}

\begin{figure}[htb]
\centering
     \includegraphics[width=0.45\textwidth]{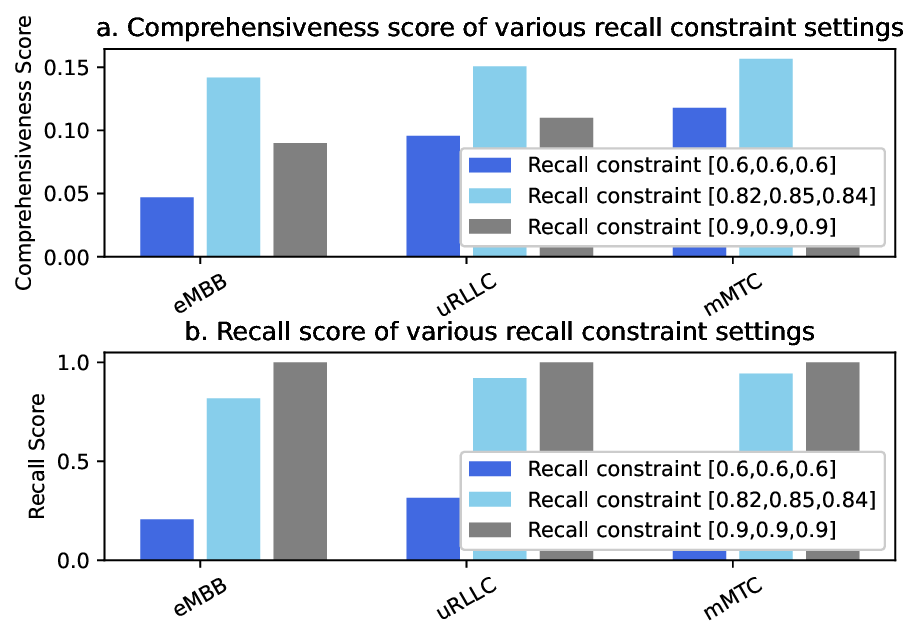}
\caption{Approach to set the lower bound of recall score of the proposed EGFL-JS}
\label{recall_set}
\end{figure}

Noted that, in our experiments, we set the lower bound recall scores for eMBB, uRLLC, and mMTC at [0.82, 0.85, 0.84], respectively, during training, aiming for values exceeding 80 $\%$. These thresholds represent a high requirement for recall that our solution successfully achieved. Our approach aligns with previous work, as evidenced by our previous study \cite{mine2}.
In real-world deployment scenarios, the recall threshold serves as a tunable parameter that can be adjusted by the operator or tenant based on the desired fairness level. To evaluate the effectiveness of these recall scores, we conducted supporting simulations, illustrating our system's performance across various recall constraints for all slices within the EGFL-JS framework. Results of fig. \ref{recall_set} show consistently higher comprehensive scores and higher recall score for all slices at the specified recall thresholds compared to alternative settings. 
Moreover, we can observe that setting recall thresholds too low risks overlooking crucial instances, leading to decreased recall and comprehensiveness scores. Conversely, overly stringent constraints may boost recall but at the expense of comprehensiveness, potentially compromising the model's generalization capabilities.
This analysis underscores the importance of selecting appropriate recall constraints to achieve desired performance outcomes.

\section{Navigating EGFL: Methodological Insights and XAI Integration Significance}

In this section, we comprehensively compare our proposed EGFL methodology and const. FL (Vanilla FL) from diverse perspectives. The aim is to underscore the importance and necessity of integrating XAI explanations within the closed-loop FL training framework. \\
As presented before,in Fig.\ref{loss_eMBB},\ref{loss_b},\ref{loss_mmtc}, EGFL exhibits faster convergence for all slices compared to const. FL (vanilla FL), highlighting its adaptability to data dynamics. Moving to Fig. \ref{recall}, recall scores for EGFL and const. FL are nearly identical across all slices, indicating consistent recall performance in EGFL despite accelerated convergence. Additionally, Fig. \ref{loss} reveals a lower training loss in EGFL, signifying improved optimization and model fit. In Fig. \ref{comp}, EGFL attains a higher comprehensiveness score for all slices compared to vanilla FL, emphasizing the benefits of XAI integration. \\
More clearly, the success of proposed EGFL is rooted in its unique design, integrating iterative local learning with runtime explanation in a closed-loop framework. A key innovation is the introduction of a novel federated EGL loss function, incorporating Jensen Shannon Divergence (JSD) to enhance model reliability. Computed using both predicted and masked predicted values generated through XAI-based feature masking (detailed in Section V), JSD incentivizes the model to generate samples resembling a target distribution. This nuanced strategy, absent in the baseline const. FL, emphasizes the synergy between model predictions and XAI-based feature masking, highlighting EGFL's distinctive contributions in the closed-looped federated learning landscape.

\section{Conclusion}

In this paper, we have proposed a novel fair explanation guided federated learning (EFGL) approach to fulfill transparent and trustworthy zero-touch service management of 6G network slices at RAN in a non-IID setup. We have proposed a modified objective function where we consider Jenson-Shannon divergence along with the cross-entropy loss and a fairness metric, namely the recall, as a constraint. The underlying optimization task has been solved using a proxy-Lagrangian two-player game strategy. The simulated results have validated that the proposed scheme can give us a new direction for implementing trustworthy and fair AI-based predictions.
\section{Acknowledgment}
This work has been supported in part by the projects 6G-BRICKS (101096954) HORIZON-JU-SNS-2022 and the 5GMediaHUB project (101016714) H2020-EU.2.1.1.
\bibliographystyle{IEEEtran}
\bibliography{myBibliographyFile}

\balance

\end{document}